\documentclass{llncs}
\usepackage{cite,microtype,graphicx}
\usepackage{hyperref,aliascnt}
\usepackage{amsmath,amssymb,amsfonts}

\usepackage{etoolbox}
{\makeatletter\patchcmd{\@spthm}{\phantomsection}{}{}{}}

\let\doendproof\endproof
\renewcommand\endproof{~\hfill$\qed$\doendproof}

\title{A Stronger Lower Bound on Parametric Minimum Spanning Trees}

\author{David Eppstein}
\institute{Computer Science Department, University of California, Irvine}

\date{ }

\begin{document}

\maketitle
\begin{abstract}
We prove that, for an undirected graph with $n$ vertices and $m$ edges, each labeled with a linear function of a parameter $\lambda$, the number of different minimum spanning trees obtained as the parameter varies can be $\Omega(m\log n)$.
\end{abstract}

\section{Introduction}

In the \emph{parametric minimum spanning tree problem}~\cite{Gus-HCGT-79}, the input is a graph $G$ whose edges are labeled with linear functions of a parameter $\lambda$. For any value of $\lambda$, one can obtain a spanning tree $T_\lambda$ as the minimum spanning tree of the weight functions, evaluated at $\lambda$.  Varying $\lambda$ continuously from $-\infty$ to $\infty$ produces in this way a discrete sequence of trees, each of which is minimum within some range of values of $\lambda$. How many different spanning trees can belong to this sequence, for a worst case graph, and how can we construct them all efficiently? Known bounds are that the number of trees in a graph with $n$ vertices and $m$ edges can be $\Omega\bigl(m\alpha(n)\bigr)$ (where $\alpha$ is the inverse Ackermann function)~\cite{Epp-DCG-98} and is always $O(mn^{1/3})$~\cite{Dey-DCG-98}; both bounds date from the 1990s and, although far apart, have not been improved since. The sequence of trees can be constructed in time $O(mn\log n)$~\cite{FerSluEpp-NJC-96} or in time $O(n^{2/3}\log^{O(1)} n)$ per tree~\cite{AgaEppGui-FOCS-98}; faster algorithms are also known for planar graphs~\cite{FerSlu-TCS-97} or for related optimization problems that construct only a single tree in the parametric sequence~\cite{KatTok-SoCG-01,Cha-SODA-05}. In this paper we improve the 25-year-old lower bound on the number of parametric minimum spanning trees from $\Omega\bigl(m\alpha(n)\bigr)$ to $\Omega(m\log n)$.

A broad class of applications of this problem involves bicriterion optimization, where each edge of a graph has two real weights of different types (say, investment cost and eventual profit) and one wishes to find a tree optimizing a nonlinear combination of the sums of these two weights (such as the ratio of total profit to total investment cost, the return on the investment). Each spanning tree of $G$ may be represented by a planar point whose Cartesian coordinates are the sums of its two kinds of weights, giving an exponentially large cloud of points, one per tree. The convex hull of this point cloud has as its vertices the parametric minimum spanning trees (and maximum spanning trees) for linear weight functions obtained from the pair of weight values on each edge by using these values as coefficients. (Essentially, this construction of weight functions from pairs of weights is a form of projective duality transforming points into lines, and the equivalence between the convex hull of the points representing trees into the lower envelope of lines representing their total weight is a standard reflection of that projective duality.)  Any bicriterion optimization problem that can be expressed as maximizing a quasiconvex function (or minimizing a quasiconcave function) of the two kinds of total weight automatically has its optimum at a convex hull vertex, and can be solved by constructing the sequence of parametric minimum spanning trees and evaluating the combination of weights for each one~\cite{Kat-IEICE-92}. Other combinatorial optimization problems that have been considered from the same parametric and bicriterion point of view include shortest paths~\cite{Car-TR-84,Eri-SODA-10,ChaFisLac-STACS-10,CasLabVio-Nw-17}, optimal subtrees of rooted trees~\cite{CarEpp-SWAT-06}, minimum-weight bases of matroids~\cite{Epp-DCG-98}, minimum-weight closures of directed graphs~\cite{Epp-TALG-18}, and the knapsack problem~\cite{Ebe-MS-96,GiuHalRuz-IPL-17,HolKru-IPL-17}.

The main idea behind our new lower bound is a recursive construction of a family of graphs (more specifically, 2-trees), formed by repeated replacement of edges by triangles (\autoref{fig:triangulate}). We also determine the parametric weight functions of these graphs by a separate recursive construction (\autoref{fig:recursion}). However, this only produces an $\Omega(n\log n)$ lower bound, because for a graph constructed in this way with $n$ vertices, the number of edges is $2n-3$, only a constant factor larger than the number of vertices. To obtain our claimed $\Omega(m\log n)$ lower bound we use an additional packing argument, in which we find a dense graph containing many copies of our sparse lower bound construction, each contributing its own subsequence of parametric minimum spanning trees to the total.

\section{Background and preliminaries}

The \emph{minimum spanning tree} of a connected undirected graph with real-valued edge weights is a tree formed as a subgraph of the given graph, having the minimum possible total edge weight. As outlined by Tarjan~\cite{Tar-83}, standard methods for constructing minimum spanning trees are based on two rules, stated most simply for the case when all edge weights are distinct. The \emph{cut rule} concerns cuts in the graph, partitions of the vertices into two subsets; an edge \emph{spans} a cut when its two endpoints are in different subsets. The cut rule states that (for distinct edge weights) the minimum-weight edge spanning any given cut in a graph belongs to its unique minimum spanning tree. The \emph{cycle rule}, on the other hand, states that (again for distinct edge weights) the maximum-weight edge in any cycle of the graph does not belong to its unique spanning tree. One consequence of these rules is that the minimum spanning tree depends only on the sorted ordering of the edge weights, rather than on more detailed properties of their numeric values.

An input to the \emph{parametric minimum spanning tree} problem consists of an undirected connected graph whose edges are labeled with linear functions of a parameter $\lambda$ rather than with real numbers. For any value of $\lambda$, plugging $\lambda$ into these functions produces a system of real weights for the edges, and therefore a minimum spanning tree $T_\lambda$. Different values of $\lambda$ may produce different trees, and the task is either to obtain a complete description of which tree is minimum for each possible value of $\lambda$ or, in some versions of the problem, to find a value $\lambda$ and its tree optimizing another objective function.

If we plot the graphs of the linear functions of a parametric minimum spanning tree instance, as lines in the $(\lambda,\text{weight})$ plane, then the geometric properties of this \emph{arrangement of lines} are closely related to the combinatorial properties of the parametric minimum spanning tree problem. If no two edges have the same weight function, then all edge weights will be distinct except at a finite set of values of $\lambda$, the $\lambda$-coordinates of points where two lines in the arrangement cross. As $\lambda$ varies continuously, the sorted ordering of the weights will remain unchanged except when $\lambda$ passes through one of these crossing points, where the set of lines involved in any crossing will reverse their weight order. It follows from these considerations that the sequence of parametric minimum spanning trees is finite, and that these trees change only at certain \emph{breakpoints} which are necessarily the $\lambda$-coordinates of crossings of lines. In particular, $m$ lines have $O(m^2)$ crossings and there can be only $O(m^2)$ distinct trees in the sequence of parametric minimum spanning trees. However, a stronger bound, $O(mn^{1/3})$, is known~\cite{Dey-DCG-98}.

The worst-case instances of the parametric minimum spanning tree problem, the ones with the most trees for their numbers of edges and vertices, have distinct edge weight functions whose arrangement of lines has only \emph{simple crossings}, crossings of exactly two lines. For, in any other instance, perturbing the edge weight functions by a small amount will preserve the ordering of weights away from the crossings of its lines, and therefore will preserve its sequence of trees away from these crossings, while only possibly increasing the number of breakpoints near perturbed crossings of multiple lines, which become multiple simple crossings. For an instance in which the lines have only simple crossings, the only possible change to the minimum spanning tree at a breakpoint is a \emph{swap}, a change to the tree in which one edge (corresponding to one of the two crossing lines at a simple crossing) is removed, and the other edge (corresponding to the other of the two crossing lines) is added in its place. For details on this correspondence between the geometry of line arrangements and the sequence of parametric minimum spanning trees, and generalizations of this correspondence to other matroids than the matroid of spanning trees, see our previous paper on this topic~\cite{Epp-DCG-98}.

\section{Replacing edges by triangles}

\begin{figure}[t]
\includegraphics[width=\textwidth]{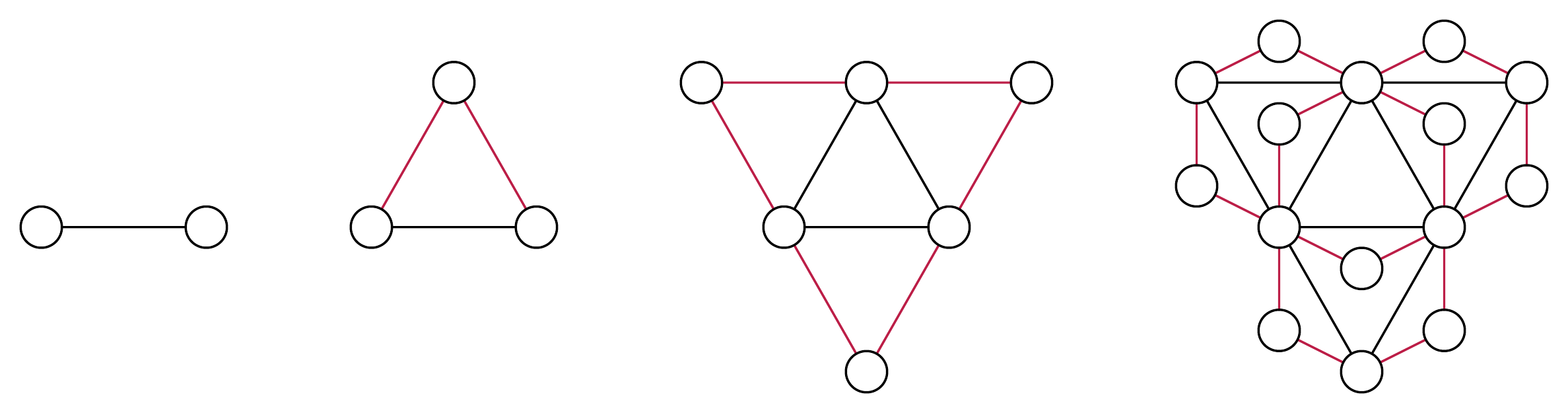}
\caption{Recursively constructing a family of 2-trees $T_i$ (here, $i=0,1,2,3$ in left-to-right order) by repeatedly replacing every edge of $T_{i-1}$ by a triangle.}
\label{fig:triangulate}
\end{figure}

A \emph{2-tree} is a graph obtained from the two-vertex one-edge graph $K_2$ by repeatedly adding new degree-two vertices, adjacent to pairs of adjacent earlier vertices. Equivalently, they are obtained by repeatedly replacing edges by triangles. These graphs are planar and include the maximal outerplanar graphs~\cite{Mit-IPL-79}; their subgraphs are the \emph{partial 2-trees}, graphs of treewidth $\le 2$~\cite{WalCol-Nw-83}. The graphs we use in our lower bound are a special case of this construction where we apply this edge replacement process simultaneously to all edges in a smaller graph of the same type. We define the first graph $T_0$ in our sequence of graphs to be the graph $K_2$, and then for all $i>0$ we define $T_i$ to be the graph obtained by replacing all edges of $T_{i-1}$ by triangles. It seems natural to call these \emph{complete 2-trees}, by analogy to complete trees (whose leaves are repeatedly replaced by stars for a given number of levels) but we have been unable to find this usage in the literature. The graphs $T_i$ for $i\le 3$ are depicted in \autoref{fig:triangulate}.

\begin{lemma}
\label{lem:t-size}
$T_i$ has $3^i$ edges and $(3^i+3)/2$ vertices.
\end{lemma}

\begin{proof}
The bound on the number of edges follows from the fact that each replacement of edges by triangles triples the number of edges. The bound on the number of vertices follows easily by induction on~$i$, using the observations that each edge of $T_{i-1}$ leads to a newly added vertex in $T_i$ and that $(3^{i-1}+3)/2+3^{i-1}=(3^i+3)/2$.
\end{proof}

What happens when we replace an edge by a triangle in a parametric spanning tree problem? For a non-parametric minimum spanning tree, the answer is given by the following lemma.

\begin{lemma}
\label{lem:triangulated-mst}
Let graph $G$ contain edge $pq$, and replace this edge by a triangle $pqr$ to form a larger graph $G^+$. Suppose that the edges in $G^+$ have distinct edge weights, and use these weights to assign weights to the edges in $G$, with the following exception: in $G$, give edge $pq$ the weight of the \emph{bottleneck} edge in triangle $pqr$ (the maximum-weight edge on path from $p$ to $q$ in the minimum spanning tree of the graph of the triangle) instead of the weight of $pq$. Then, the minimum spanning tree of $G+$ has the same set of edge weights as the minimum spanning tree of $G$, together with the minimum weight of a non-bottleneck edge in triangle $pqr$.
\end{lemma}

\begin{proof}
If $pq$ is the heaviest edge in $pqr$ then the path from $p$ to $q$ in the minimum spanning tree of $pqr$ passes through $r$, the bottleneck edge is the heavier of the two edges on this path, and the minimum non-bottleneck edge is the lighter of its two edges. Otherwise, $pq$ is the bottleneck edge and again the minimum non-bottleneck edge is the lighter of the two remaining edges incident to $r$. Applying the cut rule to the cut separating $r$ from the rest of the graph shows that the minimum non-bottleneck edge is an edge of the minimum spanning tree of $G^+$. Since we did not include its edge weight in the weights for $G$, its weight is not included in the set of edge weights of the minimum spanning tree for $G$.

Contracting this minimum non-bottleneck edge in $G^+$ produces a multigraph with two copies of edge $pq$, the lighter of which is the bottleneck edge. Therefore, if we keep only the lighter of the two edges, we obtain the weighting on $G$ as a contraction of a minimum spanning tree edge in $G^+$. This contraction preserves the set of remaining minimum spanning tree weights, as the lemma states.
\end{proof}

\begin{figure}[t]
\includegraphics[width=\textwidth]{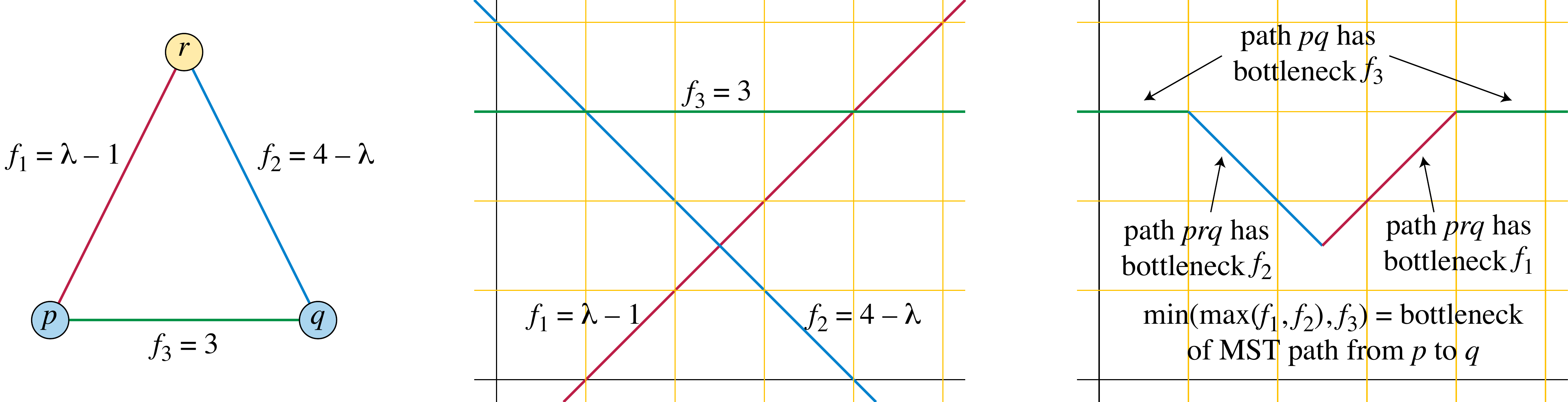}
\caption{A parametric spanning tree problem on a single triangle $pqr$, and the graph of the bottleneck edge weight on the path from $p$ to $q$ in the parametric spanning tree, as a function of the parameter $\lambda$.}
\label{fig:zigzag}
\end{figure}

It follows that in the parametric case, replacing an edge $pq$ by a triangle $pqr$, with linear parametric weights on each triangle edge, causes that edge to behave as if it has a nonlinear piecewise linear weight function attached to it, the function mapping the parameter $\lambda$ to the bottleneck weight from $p$ to $q$ in triangle $pqr$. \autoref{fig:zigzag} shows an example of three parametric weights on a triangle $pqr$ and this bottleneck weight function, with the weights chosen so that the function has three breakpoints. Clearly, we can perturb these three weight functions within small neighborhoods of their coefficients, and obtain a qualitatively similar bottleneck weight function.

\section{Weighted 2-trees}

\begin{figure}[t]
\includegraphics[width=\textwidth]{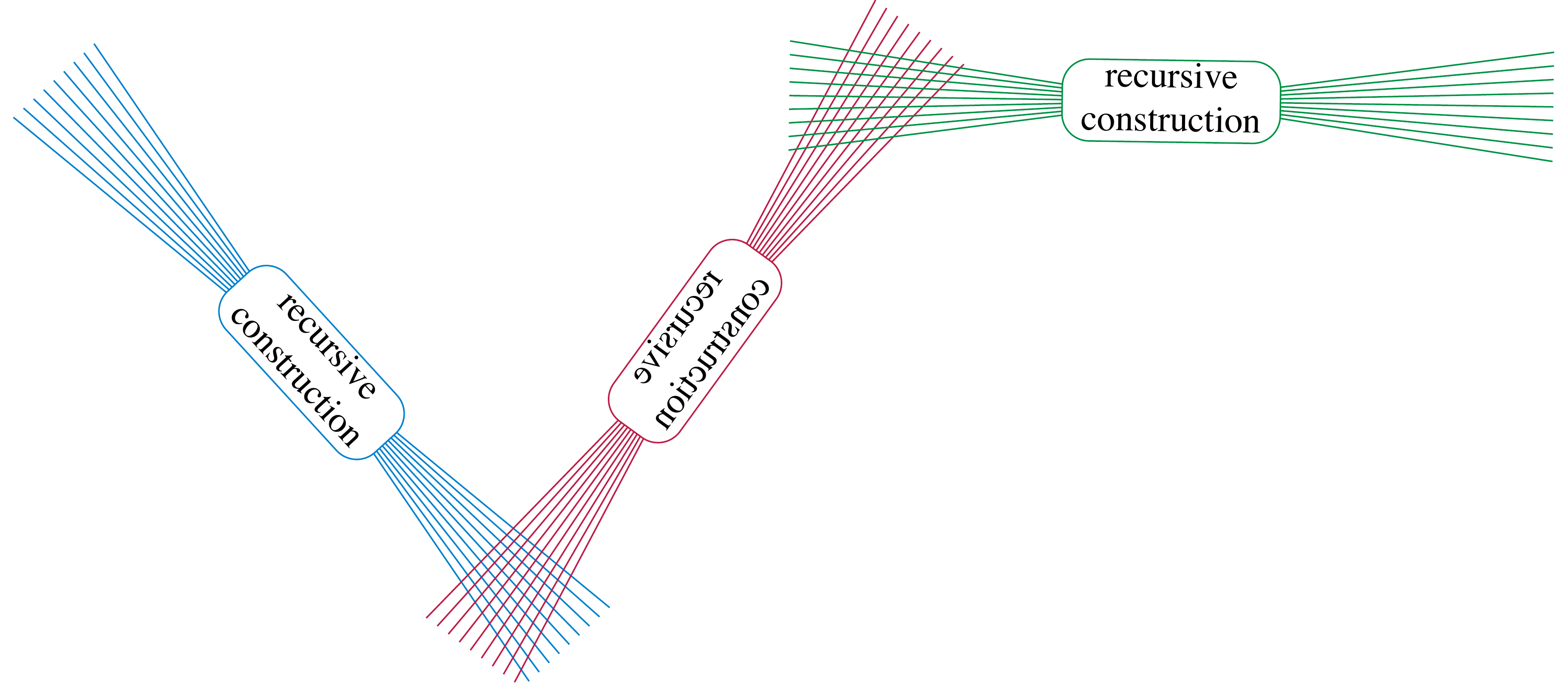}
\caption{Recursive construction for the parametric weight functions of the graphs $T_i$, shown here as an arrangement of lines in a plane whose horizontal coordinate is the parameter $\lambda$ and whose vertical coordinate is the edge weight at that parameter value. The reversed text in the central recursive construction indicates that the construction is reversed left-to-right relative to the other two copies.}
\label{fig:recursion}
\end{figure}

We now describe how to assign parametric weights to the edges of $T_i$ to obtain our $\Omega(n\log n)$ lower bound.
As a base case, we may use any linear function as the weight of the single edge of $T_0$; it can have only one spanning tree, regardless of this choice. For $T_i$, with $i>0$, we perform the following steps to assign its weights:
\begin{itemize}
\item Construct the weight functions for the edges of $T_{i-1}$, recursively.
\item Apply a linear transformation to the parameter of these weight functions (the same transformation for each edge) so that, in the arrangement of lines representing the graphs of these weight functions, all crossings occur in the interval $[0,1]$ of $\lambda$-coordinates. Additionally, scale these weight functions by a sufficiently small factor $\epsilon$ so that, within this interval, they are close enough to the $\lambda$-axis, for a meaning of ``close enough'' to be specified below.
\item Construct $T_i$ by replacing each edge $pq$ in $T_{i-1}$ by a triangle $pqr$, with a new vertex for each triangle. Color the three edges of each triangle red, blue, and green, as in \autoref{fig:zigzag}(left), with $pq$ colored green and the other two edges colored red and blue (choosing arbitrarily which one to color red and which one to color blue).
\item Give each edge of $T_i$ a transformed copy of the weight function of the corresponding edge of $T_{i-1}$, transformed as follows:
\begin{itemize}
\item For a green edge $pq$, corresponding to an edge of $T_{i-1}$ with weight function $f(\lambda)$, give $pq$ the weight function $f(\lambda-4.5)+3$. This transformation shifts the part of the weight function where the crossings with other green edges occur to be close to the right green segment of \autoref{fig:zigzag}(right).
\item For a red edge $pr$, corresponding to an edge $pq$ of $T_{i-1}$ with weight function $f(\lambda)$, give $pr$ the weight function $f(3.75-\lambda)+\lambda-1$. This transformation shifts the part of the weight function where the crossings with other red edges occur to be close to the red segment of \autoref{fig:zigzag}(right), and (by negating $\lambda$ in the argument to $f$) reverses the ordering of the crossings within that region.
\item For a blue edge $qr$, corresponding to an edge $pq$ of $T_{i-1}$ with weight function $f(\lambda)$, give $qr$ the weight function $f(\lambda-1.25)+4-\lambda$. This transformation shifts the part of the weight function where the crossings with other red edges occur to be close to the blue segment of \autoref{fig:zigzag}(right).
\end{itemize}
\item Perturb all of the weight functions, if necessary, so that all crossings of two weight functions have different $\lambda$-coordinates, without changing the left-to-right ordering of the crossings between any one weight function and the rest of them.
\end{itemize}

This construction is depicted schematically, in the $(\lambda,\text{weight})$ plane, in \autoref{fig:recursion}. We are now ready to define what it means for the weight scaling factor $\epsilon$ to be small enough, so that the scaled weight functions are ``close enough'' to the $\lambda$ axis: as shown in the figure, the left-to-right ordering of the crossings of the lines graphing the weight functions should be:
\begin{enumerate}
\item All crossings of blue with green lines
\item All crossings of two blue lines, in one copy of the recursive construction
\item All crossings of blue with red lines
\item All crossings of two red lines, in a second (reversed) copy of the recursive construction
\item All crossings of red with green lines
\item All crossings of two green lines, in the third copy of the recursive construction
\end{enumerate}
Our construction automatically places all monochromatic crossings into disjoint unit-length intervals with these orderings.
The bichromatic crossings of \autoref{fig:zigzag} are separated from these unit-length intervals by a horizontal distance of at least $0.25$, and sufficiently small values of $\epsilon$ will cause the bichromatic crossings of $T_i$ to be close to the positions of the crossings with the same color in \autoref{fig:zigzag}. Therefore, by choosing $\epsilon$ small enough, we can ensure that the crossing ordering described above is obtained. \autoref{fig:sun} depicts this construction for $T_2$.

We observe that, within each of the unit-length intervals containing a copy of the recursive construction,
the bottleneck edges for each triangle $pqr$ in the construction of $T_i$ are exactly the ones of the color for that copy of the recursive construction, and that within these intervals, the minimum non-bottleneck edge in each triangle does not change. Therefore, by \autoref{lem:triangulated-mst}, the changes in the sequence of parametric minimum spanning trees within these intervals exactly correspond to the changes in the trees of $T_{i-1}$ from the recursive construction.

\begin{figure}[t]
\includegraphics[width=\textwidth]{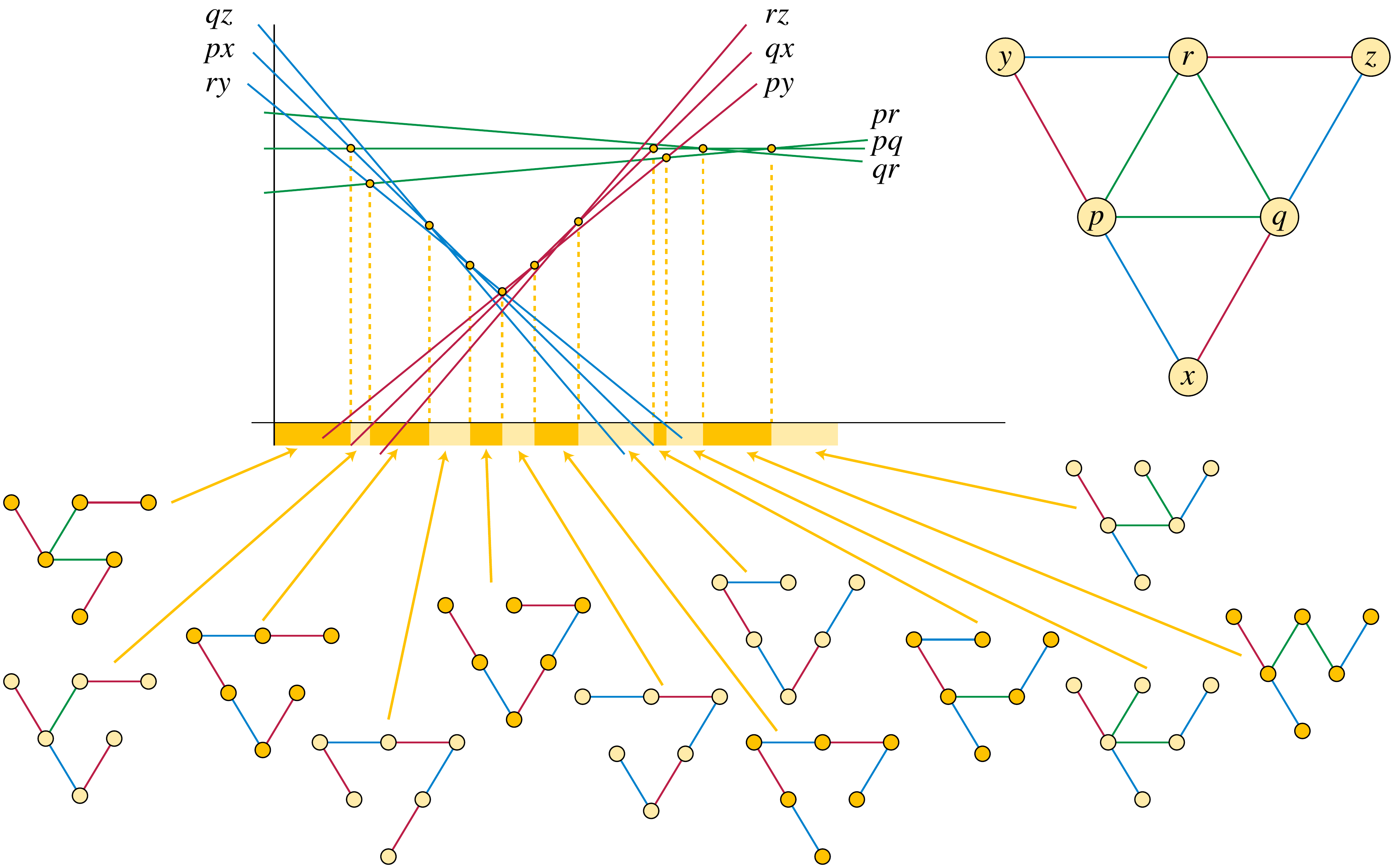}
\caption{$T_2$ (upper right) as parametrically weighted in our construction, with the graphs of each weight function shown as lines in the $(\lambda,w)$ plane (upper right), and the resulting sequence of 12 parametric minimum spanning trees (bottom). The marked yellow crossings of pairs of lines correspond to breakpoints in the sequence of trees.}
\label{fig:sun}
\end{figure}

\begin{lemma}
\label{lem:recurrence}
For weights constructed as above, the number of distinct parametric minimum spanning trees for $T_i$ is at least as large as
\[
N(i)=\frac{i3^i}{2}+\frac{3^i+3}{4}.
\]
\end{lemma}

\begin{proof}
We prove by induction on $i$ that the number of trees is at least as large as the solution to the recurrence
\[
N(i) = 3N(i-1) + \frac{3^i-3}{2}.
\]
To prove this, it is easier to count the number of breakpoints, values of $\lambda$ at which the tree structure changes; the number of trees is the number of breakpoints plus one. In each copy of the recursive construction, this number of breakpoints is exactly $N(i-1)-1$, so the total number of breakpoints appearing in these three copies is $3N(i-1)-3$.

Additional breakpoints happen within the ranges of values for $\lambda$ at which (in the $(\lambda,\text{weight})$ plane) pairs of lines of two different colors cross. Because of the reversal of the red copy of the recursive construction, the minimum spanning trees immediately to the left and right of  these regions of bichromatic crossings correspond to the same trees in $T_{i-1}$: the bottleneck edges that are included in these minimum spanning trees come from the same triangles, but with different colors.  In the regions where the green lines cross lines of other colors, the minimum non-bottleneck edge in each triangle does not change, so each green bottleneck edge in the minimum spanning tree must be exchanged for a red or blue one. Each change to a tree within this crossing region removes a single edge from the minimum spanning tree and replaces it with another single edge, the two edges whose two lines cross at the $\lambda$-coordinate of that change. Therefore, no matter what sequence of changes is performed, to exchange all green bottleneck edges for all red or blue ones requires a number of crossings equal to the number of edges in the minimum spanning tree of $T_{i-1}$, which is $(3^{i-1}+1)/2$ by \autoref{lem:t-size}. We get this number of breakpoints at the region where the green and blue lines cross, and the same number at the region where the red and green lines cross.

The analysis of the number of breakpoints at the region where the blue and red lines cross is similar, but slightly different.
Immediately to the left and right of this region, the the bottleneck edge in each triangle and the minimum non-bottleneck edge in the triangle are red and blue, but in a different order to the left and to the right. Therefore, in triangles where the bottleneck edge is part of the minimum spanning tree (as is always the minimum non-bottleneck edge), nothing changes. However, in triangles where the bottleneck edge is \emph{not} part of the minimum spanning tree, there is a change, to the minimum non-bottleneck edge, from before this crossing region to after it. These triangles correspond to edges of $T_{i-1}$ which do not belong to the minimum spanning tree (for the parameter values in this range), of which there are $(3^{i-1}-1)/2$ by \autoref{lem:t-size}. By the same argument as before, the crossing region must contain at least this many breakpoints.

Adding together the $3N(i-1)-3$ breakpoints from the recursive copies, the $(3^{i-1}+1)/2$ breakpoints for the green--red and green--blue crossing regions,  the $(3^{i-1}-1)/2$ breakpoints for the red--blue crossing region, and $+1$ to convert numbers of breakpoints to numbers of distinct trees, and simplifying, gives the right hand side of the recurrence. A straightforward induction shows that the solution to the recurrence is the formula given in the statement of the lemma.
\end{proof}

For $i=0,1,2,\dots$ the number of trees given by this formula is
\[
1, 3, 12, 48, 183, 669, 2370, 8202, 27885, 93495 \dots
\]
For instance, $T_1$ has three trees with the weighting given in \autoref{fig:zigzag}: the bottleneck function shown in the figure has four linear pieces, but the red and blue pieces both correspond to the same tree, with a different edge on the path $pqr$ as the bottleneck edge. \autoref{fig:sun} shows the 12 trees for $T_2$.

\section{Packing into dense graphs}

The lower bound obtained from \autoref{lem:recurrence} applies only to sparse graphs, where the numbers of vertices and edges are within constant factors of each other. However, we want a bound that applies more generally, for graphs with significantly more edges than vertices. The other direction, for graphs with significantly fewer edges than vertices, is less interesting. To achieve many fewer edges than vertices, it is necessary to allow disconnected graphs, and consider minimum spanning forests instead of minimum spanning trees; but with these modifications one can obtain a lower bound simply by adding isolated vertices to the construction of \autoref{lem:recurrence}.

To achieve many more edges than vertices, we use the following construction for packing many instances of a sparse lower bound graph into a single denser graph. It does not require any detailed knowledge of the structure of the sparse graph.

\begin{figure}[t]
\includegraphics[width=\textwidth]{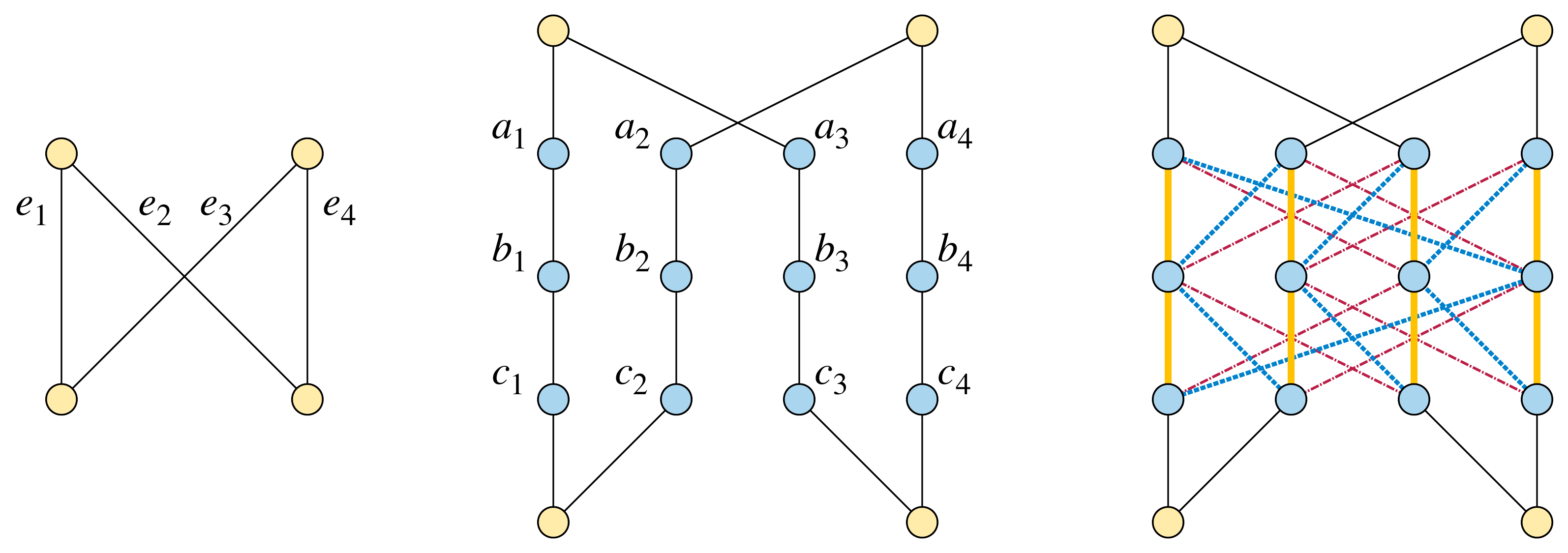}
\caption{The construction of \autoref{lem:pack}, applied to a graph $G$ with four vertices and four edges (left), with the parameter $k=3$. The central graph is a subdivision of each edge of this graph into a four-edge path, with vertices labeled as shown, and the graph on the right is the final construction $H$, with the colors and textures of edges indicating the partition of its edges into four subgraphs $H_0$ (thin black edges), $H_1$ (thick yellow edges), $H_2$ (dotted blue edges), and $H_3$ (dashed red edges).}
\label{fig:pack}
\end{figure}

\begin{lemma}
\label{lem:pack}
Let $G$ be a parametrically weighted graph with $N$ vertices and $M$ edges, whose sequence of parametric minimum spanning trees has length $T$, and let $k$ be a positive integer satisfying $k\le M$. Then there is a parametrically weighted graph $H$ with $N+3M$ vertices and $(2k+2)M$ edges whose sequence of parametric minimum spanning trees has length at least $2kT$.
\end{lemma}

\begin{proof}
We construct $H$ from $G$ in the following steps, illustrated in \autoref{fig:pack}.
\begin{itemize}
\item Number the edges of $G$ as $e_0, e_2,\dots e_{M-1}$ arbitrarily.
\item Subdivide each edge $e_i$ of $G$, connecting two vertices $u$ and $v$, into a four-edge path $u$--$a_i$--$b_i$--$c_i$--$v$. (It is arbitrary which vertex of this path we call $a_i$ and which we call $c_i$.)
\item Add additional edges from $b_i$ to $a_j$ and $c_j$, for each $i$ and each $j=i+1,i+2,\dots i+k-1\bmod{M}$.
\end{itemize}
Given this construction, we define subgraphs $H_j$ as follows:
\begin{itemize}
\item $H_0$ consists of all edges connecting vertices of $G$ to new vertices $a_i$ or $c_i$.
\item $H_j$ consists of all edges from $b_i$ to $a_{i+j-1}$ or $c_{i+j-1}$, for all $i$, with indexes taken modulo~$m$.
\end{itemize}
Then, for $i=1,2,\dots k$, the graph $H_0\cup H_i$ is isomorphic to a subdivision of $G$, with $H_0\cup H_1$ being the subdivision we used to construct $H$ and the others obtained in the same way but with permuted connections.

\begin{figure}[t]
\includegraphics[width=\textwidth]{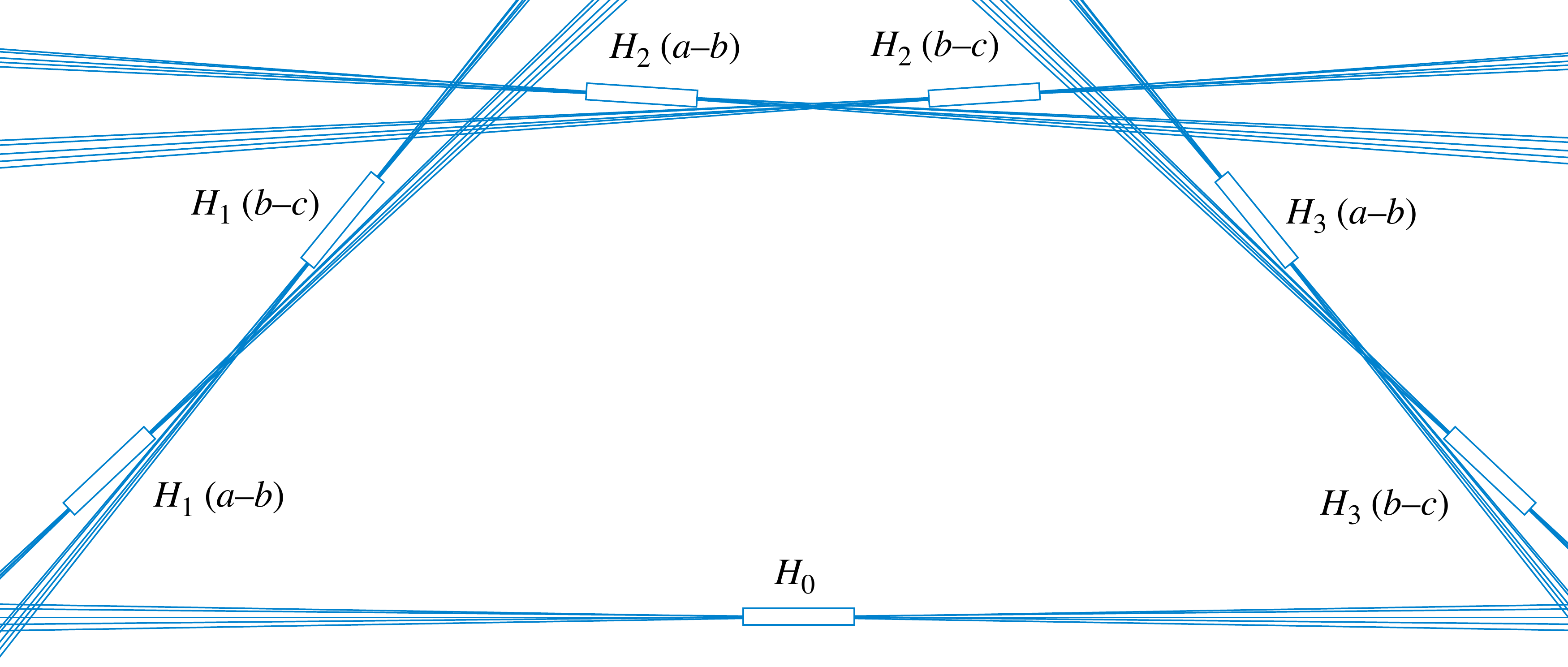}
\caption{An arrangement of lines for the weight functions of \autoref{lem:pack} with $k=3$. The small rectangles indicate transformed neighborhoods of the unit $\lambda$-interval, containing all crossings of the bundle of lines associated with each subgraph.}
\label{fig:dense-weights}
\end{figure}

As in \autoref{lem:recurrence}, we flatten the arrangement of lines for the weighting of $G$ so that its crossings all lie within a small neighborhood of the unit interval of the $\lambda$-axis, without changing its sequence of parametric minimum spanning trees. We then apply linear transformations to the system of weights for the edges in each copy $H_j$ with $j>0$, as detailed below, while using small-enough weights for all edges in $H_0$ so that these edges belong to all minimum spanning trees for parameters in the range covered by the transformed unit intervals. shown in \autoref{fig:dense-weights}. More specifically, for each $j>0$ we use one transformed copy of the weights in $G$ for the $a$--$b$ edges in $H_j$, and a second transformed copy for the $b$--$c$ edges, arranged so that the transformed unit intervals containing the crossings within each copy project to disjoint intervals of the $\lambda$-axis, and so that all crossings of the $a$--$b$ edges appear above all lines for the $b$--$c$ edges and vice versa. Therefore, in the graph $H_0\cup H_j$, the parametric trees in the parameter range where the $a$--$b$ edges cross each other consist of all $b$--$c$ edges (because those have smaller weight than the $a$--$b$ edges in each path) together with a subset of the $a$--$b$ edges corresponding to a spanning tree of $G$. Because we copied and transformed the weights of $G$ for the $a$--$b$ edges in this parameter range, we obtain $T$ distinct trees of this type. To arrange the $a$--$b$ and $b$--$c$ parameter weights for $H_i$ in this fashion, we transform them so that the $a$--$b$ weights lie near the line $w=3-\lambda$ , with crossings in the range $\lambda\in[1,2]$, and so that the $b$--$c$ weights lie near the line $w=\lambda-3$, with crossings in the range $\lambda\in[4,5]$. Then, we transform and flatten these combined weights of $H_i$, so that they again lie near the $\lambda$-axis with all crossings of edges of either type in the range $[0,1]$.

We arrange the sets of lines associated with $H_1$, $H_2$, etc., so that the lines from each $H_j$ pass above the crossings for each other $H_j'$, $j\ne j'$, and so that the range of parameters within which $H_j$ has the lowest lines contains the two subranges where its $a$--$b$ lines cross and where its $b$--$c$ lines cross, again
as shown in the figure. We may do this by finding a convex-downward polygonal chain with $k$ sides (for instance the upper part of a regular $2k$-gon), in which all sides project to a range of $\lambda$-coordinates of more than unit length, and by transforming the weights of each $H_i$ so that the unit interval of the $\lambda$-axis, near which all crossings of these weights occur, is transformed to the interior of one of the sides of this polygonal chain. \autoref{fig:dense-weights} shows the weights for three subgraphs $H_1$, $H_2$, and $H_3$, transformed in this way so that they are near the upper three sides of a hexagon. The weights for $H_0$ can be chosen to be near a horizontal line, below all crossings of the other weight functions, as also shown in the figure.

Therefore, within these subranges, the parametric minimum spanning trees for all of $H$ will be the same as the trees for $H_0\cup H_j$, because $H_0\cup H_j$ spans $H$ and has lower edge weights than any of the remaining edges. With this arrangement, we get $2kT$ distinct parametric minimum spanning trees, $2T$ for each $H_j$ with $j>0$, as well as additional trees that are not counted in the lemma.
\end{proof}

With this, we are ready to prove our main result:

\begin{theorem}
There exists a constant $C$ such that the following is true.
Let $n$ and $m$ be integers with $n>0$ and $2n-3\le m\le\tbinom{m}{2}$.
Then there exists a parametrically weighted graph with $n$ vertices and $m$ edges,
with at least $Cm\log n$ parametric minimum spanning trees.
\end{theorem}

\begin{proof}
Let $G=T_i$, $N=(3^i+3)/2$, and $M=3^i$, with $i$ chosen as large as possible so that $N+3M\le n$ and $4M\le m$, and choose $k$ as large as possible so that $(2k+2)M\le m$; then $N=\Theta(n)$ and $M=\Theta(m/n)$.
Apply \autoref{lem:recurrence} to give weights to $G$ so that it has $\Omega(n\log n)$ parametric minimum spanning trees, and apply \autoref{lem:pack} to construct a parametrically weighted graph $H$ with $N+3M$ vertices and $(2k+2)M$ edges that has $\Omega(m\log n)$ parametric minimum spanning trees. If necessary, add leaf vertices to $H$ to increase its number of vertices to $n$, and then add high-weight edges to increase its number of edges to $m$ without affecting this sequence of parametric spanning trees.
\end{proof}

\section{Conclusions}

We have shown that the number of parametric minimum spanning trees can be $\Omega(m\log n)$ in the worst case, improving a 25-year-old $\Omega\bigl(m\alpha(n)\bigr)$ lower bound. Because of the structure of the graphs used in our lower bound construction, the new lower bound applies as well to the special cases of planar graphs and of bounded-treewidth graphs, both of which can have $\Omega(n\log n)$ parametric minimum spanning trees. However, our new lower bound is still far from the $O(mn^{1/3})$ upper bound, so there is plenty of room for additional improvement.

Another related question concerns the \emph{parametric bottleneck shortest path problem}, a parametric version of the problem of finding a path between two specified vertices that minimizes the maximum edge weight on the path. In the non-parametric version of the problem, a minimum spanning tree path is an optimal path, although faster algorithms are possible and the problem is also of interest in the case of directed graphs~\cite{GabTar-Algs-88}. The same problem is also known in the equivalent maximin form as the widest path problem, where an optimal solution can be found as a maximum spanning tree path~\cite{Pol-OR-60}. The parametric versions of these problems differ somewhat: a breakpoint in the piecewise linear parametric minimum spanning tree function (the function mapping the parameter value $\lambda$ to the weight of its minimum spanning tree) might not be a breakpoint in the bottleneck shortest path problem (the maximum weight of an edge on the bottleneck shortest path problem) or vice versa. However, the bottleneck breakpoints that look locally like the minimum of two linear functions do correspond to breakpoints of the minimum spanning tree problem. For this reason, any asymptotic lower bound on the parametric bottleneck shortest path problem would also be a lower bound for parametric minimum spanning trees, and any asymptotic upper bound on the parametric minimum spanning tree problem (including the known $O(mn^{1/3})$ bound) is also an upper bound on parametric bottleneck shortest paths. In fact, our previous $\Omega\bigl(m\alpha(n)\bigr)$ lower bound also applies to parametric bottleneck shortest paths, but our new $\Omega(m\log n)$ bound does not. Can we strengthen the $\Omega\bigl(m\alpha(n)\bigr)$ bound for this problem?

\bibliographystyle{plainurl}
\bibliography{parametric}
\end{document}